\documentclass[11pt]{article}

\usepackage{amssymb,amsmath,amsthm,latexsym,amscd,amsfonts,mathrsfs}
\usepackage{authblk}
\usepackage{graphics}
\usepackage{cite}
\usepackage{caption}
\usepackage{epsfig}
\usepackage{epstopdf}
\usepackage{bbm}
\usepackage{dsfont}
\usepackage{setspace}
\usepackage{float}
\usepackage{relsize}
\usepackage[titletoc,toc]{appendix}
\usepackage{tikz}
\usetikzlibrary{%
  arrows,%
  shapes.misc,
  shapes.arrows,%
  chains,%
  matrix,%
  positioning,
  scopes,%
  decorations.pathmorphing,
  shapes.geometric,
  shadows,
  decorations.pathreplacing,
  patterns,
  shapes,
  decorations.markings
}

\usepackage{caption}
\tikzset{radiation/.style={{decorate,decoration={expanding waves,angle=45,segment length=4pt}}}}

\makeatletter
\newtheorem*{rep@theorem}{\rep@title}
\newcommand{\newreptheorem}[2]{%
\newenvironment{rep#1}[1]{%
 \def\rep@title{#2 \ref{##1}}%
 \begin{rep@theorem}}%
 {\end{rep@theorem}}}
\makeatother

\newtheorem{theorem}{Theorem}

\newtheorem{remark}{Remark}

\def\VR{\kern-\arraycolsep\strut\vrule &\kern-\arraycolsep}
\def\vr{\kern-\arraycolsep & \kern-\arraycolsep}

\def\mv{{\mathcal V}}
\def\ms{{\mathcal S}}
\def\e{\mathbb{E}}

\def\styp{\mathcal{A}_{\epsilon}^n}

\newreptheorem{theorem}{Theorem}

\begin{document}
%
\title{On the Equivalency of Reliability and Security Metrics for Wireline Networks}


\author[]{Mohammad Mahdi Mojahedian}
\author[]{Amin Gohari}
\author[]{Mohammad Reza Aref\thanks{This work was partially supported by Iran National Science Foundation (INSF) under contract No. 92/32575.}}

\affil[]{\footnotesize Information Systems and Security Lab. (ISSL), Sharif University of Technology, Tehran, Iran
\\m\_mojahedian@ee.sharif.edu, aminzadeh@sharif.edu, aref@sharif.edu}


\allowdisplaybreaks
\renewcommand\Authands{ and }


%


\date{}

\maketitle

\begin{abstract}
In this paper, we show the equivalency of weak and strong secrecy conditions for a large class of   secure network coding problems. When we restrict to linear operations, we show the equivalency of ``perfect secrecy and zero-error constraints" with ``weak secrecy and $\epsilon$-error constraints". 
\end{abstract}


%

\section{Introduction}
\label{sec:intro}

Reliable and secure communication requires low error probability and low information leakage. But there are different metrics for error probability and information leakage (such as weak, strong, or perfect secrecy). 
Two important reliability metrics are $\epsilon$ or zero probability of error. An $\epsilon$-error criterion requires the (average or maximal) error probability to vanish as the blocklength increases, while a zero-error criterion, demands the error to be exactly zero for every given bloklength. Three important security metrics are weak, strong, or perfect secrecy. A weak notion of secrecy requires the \emph{percentage} of the message that is leaked to vanish as the code blocklength increases, while a strong notion of secrecy requires the \emph{total amount} of leaked information (not its percentage) to vanish as the blocklength increases. Perfect secrecy requires absolutely \emph{zero} leakage of information, for every given bloklength.

These reliability and security metrics lead to different notions of capacity which could be quite different. For instance, zero-error capacity, which was originally introduced by Shannon \cite{shannon56}, could be zero in a point-to-point channel, while the $\epsilon$-error could be non-zero for the same channel. One can then ask ``how capacity behaves under different reliability and security metrics?"
There are some previous works that address this interesting question. In \cite{laneff11,changr10}, the authors showed that in the network coding problem with co-located sources, the rate region does not increase by relaxing zero-error to $\epsilon$-error condition. Maurer et al. in \cite{mawo00} proved the rate region equivalency of weak and strong secure conditions in the source model secret key agreement problem. In \cite{MoGoAr15}, the equivalency of weak and perfect secrecy conditions (with $\epsilon$-error probability) for the secure index coding problem is shown. Moreover, it is shown that zero-error probability could be achieved at the cost of a multiplicative constant. To the best of our knowledge, no other work except \cite{MoGoAr15} has concentrated on the equivalency of weak and perfect secrecy. But the setup of this problem, reviewed in Fig.~\ref{fig:Secure_IC},  is restricted. For instance, the adversary is assumed to have full access to the communication links and the shared keys are either shared between pairs of nodes, or all of the nodes (no key is shared between subsets of size three for instance). Furthermore, the network topology of index coding is a special case of wireline networks. While there are many works addressing the security aspects of wireline networks \cite{CaiYe02,ChengYe14,CuiTr10,CuiTr13,CzapFr15,RouayhebSo12,FeldmanMa04,HuangTr13,MishraFr13,SilvaKs08,fragouli16,dau2011secure} in various settings, as far as we know, none of the works in the literature study how the secrecy region changes with different criteria in secrecy constraints in the secure network coding problem. Nonetheless, important aspects of secure communication such as secure throughput in the presence of an active adversary who can corrupt a limited number of links has been considered. For more details about the works in the secure network coding problem, one can refer to \cite{fragouli16}.

\begin{figure}[t]
\centering
\resizebox {.7\columnwidth} {!} {
\begin{tikzpicture}[scale=1]

\draw[thick] (-.5,0) rectangle (2,1) node [pos=.5] {$C=f(\mathbf{M},\mathbf{K})$};
\draw[thick,->] (.75,1.5) -- (.75,1) node [pos=-.5] {$\mathbf{K}=\{K,K_1,\cdots,K_t\}$};

\draw[thick,->] (-1,.5) -- (-.5,.5) node at (-2.75,.5) {$\mathbf{M}=\{M_1,\cdots,M_t\}$};

\draw[thick,dashed] (2,.5) -- (3.5,.5);
\draw[thick,dashed] (3.5,-3) -- (3.5,3.5);
\draw[thick,dashed,->] (3.5,3.5) -- (4.5,3.5);
\draw[thick,dashed,->] (3.5,-3) -- (4.5,-3);
\draw[thick,dashed,->] (3.5,1) -- (4.5,1);

\draw[thick,dashed] (2.75,.5) -- (2.75,-2);
\draw[thick,dashed,->] (2.75,-2) -- (2,-2) node at (.75,-2) [red] {Eavesdropper};

\draw[thick,dotted] (6,0) -- (6,-1);

\draw[thick] (4.5,3) rectangle (7.5,4) node [pos=.5] {$g_1(C,\mathbf{S}_1,K,K_1)$};
\draw[thick,->] (7.5,3.5) -- (8,3.5) node [right] {$M_1$};
\draw[thick,->] (6,4.5) -- (6,4) node [pos=-.5] {$K,K_1,\mathbf{S}_1$};

\draw[thick] (4.5,.5) rectangle (7.5,1.5) node [pos=.5] {$g_2(C,\mathbf{S}_2,K,K_2)$};
\draw[thick,->] (7.5,1) -- (8,1) node [right] {$M_2$};
\draw[thick,->] (6,2) -- (6,1.5) node [pos=-.5] {$K,K_2,\mathbf{S}_2$};

\draw[thick] (4.5,-3.5) rectangle (7.5,-2.5) node [pos=.5] {$g_t(C,\mathbf{S}_t,K,K_t)$};
\draw[thick,->] (7.5,-3) -- (8,-3) node [right] {$M_t$};
\draw[thick,->] (6,-2) -- (6,-2.5) node [pos=-.5] {$K,K_t,\mathbf{S}_t$};
\end{tikzpicture}
}
\caption{\footnotesize{The schematic of 
perfectly secure index coding problem. This is a generalization of Shannon's cypher system \cite{shannon49} to an index coding setup, which was introduced by Birk and Kol \cite{biko98} in the context of satellite communication and studied further in  \cite{lust09,AlonLub2008,bazi11,tesa12,blkl13,blkl10,biko98,arbban13,shka14,nete12}. In the secure index coding problem, there is a transmitter sending $t$ messages $M_1,M_2,\cdots,M_t$ to the $t$ legitimate receivers in the presence of an eavesdropper. Each receiver $i,~i\in[t]$ has a side information set $\mathbf{S}_i$ which is a subset of messages $\{M_1,M_2,\cdots,M_t\}$ except $M_i$. Furthermore, there is a common key $K$ shared among all the legitimate parties, and private keys $K_1,K_2,\cdots,K_t$ shared between the transmitter and each of the receivers. The transmitter applies a (randomized) function on the messages and keys to compute the public code $C$. Then, $C$ is broadcast, and all the receivers including the eavesdropper can hear $C$. Each receiver $i$ applies a function on the information available to it, namely $K$, $K_i$ and messages in $\mathbf{S}_i$ to compute $M_i$. The goal is to find the minimum number of information bits that should be broadcast by the server so that each client can recover its desired messages with \emph{zero-error} probability, and further, eavesdropper could not retrieve any information about the messages by having $C$ (\emph{perfect secrecy}).}
}
\label{fig:Secure_IC}
\end{figure}
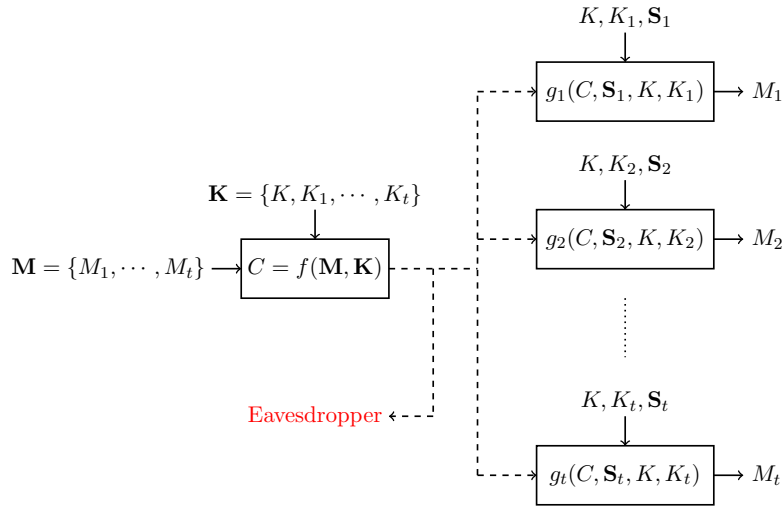

\textbf{Our contribution:}
 In this paper, we consider a general wireline network consisting of sources, intermediate nodes, and sinks, which are interconnected by error-free links. The links are directional with given capacities. Thus, wireline network can be represented by a directed weighted graph. This graph is allowed to have directed cycles. The source nodes have messages that are desired by sink nodes. 
Moreover, nodes in the network have access to infinite private randomness (only available to the nodes themselves), and also a number of rate-limited shared keys. Each key is shared among a subset of source, relay or destination nodes. These secret keys are helpful in hiding the messages from an eavesdropper who has access to a subset of links. 

Our main result is to show that changing weak to perfect condition and $\epsilon$-error to zero-error constraint, does not affect the achievable secure rate region of linear network coding  (if nodes are restricted to linear operations). When the nodes are allowed to do non-linear operations, we show that weak and strong secrecy are equivalent. 

\textbf{Notation:} Random variables are denoted by capital letters and their values by lowercase letters. We use $[k]$ to denote the set $\{1,2,\dots, k\}$. For a given subset $\mathcal{S}\subset [t]$ and a sequence of random variables $M_1,M_2,\cdots,M_t$, we use $M_\mathcal{S}$ to also denote the set $\{M_i: i\in \mathcal{S}\}$. When $\mathcal{S}=[t]$ is the full set, instead of $M_{[t]}$ we also use bold font to denote full sets, or its vector form, \emph{i.e.,} we use $\mathbf{M}$ to either denote the message set $\{M_1,M_2,\cdots,M_t\}$, or the vector $\begin{bmatrix}
M_1,M_2,\cdots,M_t
\end{bmatrix}$. Whether $\mathbf{M}$ is a set or a vector is clarified in the context. The total variation distance between two pmfs $p_X$ and $q_X$ is defined as
$$\|p_X-q_X\|_1=\frac 12\sum_x|p_X(x)-q_X(x)|.$$
We use $\mathbbm{1}[\cdot]$ to denote the indicator function; it is equal to one if the condition inside $[\cdot]$ holds; otherwise it is zero. Finally, all the logarithms in this paper are in base two.

\section{Definitions}
\label{sec:definition}

We assume that there are $t$ messages $M_1,M_2,\cdots,M_t$. Let us denote the set of all messages by $\mathbf{M}=\{M_1,M_2,\cdots,M_t\}$. 
As shown in the Fig.~\ref{fig:Randomized_Network}, the wireline network we consider in this paper consists of source nodes, receiver nodes (sink nodes) and some intermediate relay nodes. 
The nodes (source, sink and intermediate nodes) are interconnected by error-free point-to-point links. In addition, there exists an eavesdropper who is able to hear some of the links. 
Each source node has access to a subset of messages.  Similarly, each sink node desires to obtain a subset of messages. The source and sink nodes are part of the wireline network.

\begin{figure}
\centering
\resizebox {.7\columnwidth} {!} {
\begin{tikzpicture}[scale=1] 
\tikzset{->-/.style={decoration={
  markings,
  mark=at position #1 with {\arrow{>}}},postaction={decorate}}}
  
\draw[draw=none,rounded corners,dashed,fill=red!40!white, fill opacity=0.3] (-3,-2) rectangle (3,2);

\node [red] at (0,2.5) {Intermediate nodes and error-free directed links};

\draw[rounded corners,draw=none,fill=blue!40!white, fill opacity=0.3] (-4.3,-2) rectangle (-3.7,2);

\draw [fill=black] (-4,1.5) circle (.05cm) node [below] {$\mathrm{S}_1$};
\draw [->-=.5] (-4,1.5) -- (-3,1.5);

\draw [fill=black] (-4,.75) circle (.05cm) node [below] {$\mathrm{S}_2$};
\draw [->-=.5] (-4,.75) -- (-3,.75);

\draw [fill=black] (-4,-1.5) circle (.05cm) node [below] {$\mathrm{S}_t$};
\draw [->-=.5] (-4,-1.5) -- (-3,-1.5);

\draw [densely dotted] (-4,-1) -- (-4,0);

\node [blue] at (-4,-2.5) {Sources};

\draw[rounded corners,draw=none,fill=blue!40!white, fill opacity=0.3] (3.7,-2) rectangle (4.3,2);

\draw [fill=black] (4,1.5) circle (.05cm) node [below] {$\mathrm{D}_1$};
\draw [->-=.5] (3,1.5) -- (4,1.5);

\draw [fill=black] (4,.75) circle (.05cm) node [below] {$\mathrm{D}_2$};
\draw [->-=.5] (3,.75) -- (4,.75);

\draw [fill=black] (4,-1.5) circle (.05cm) node [below] {$\mathrm{D}_u$};
\draw [->-=.5] (3,-1.5) -- (4,-1.5);

\draw [densely dotted] (4,-1) -- (4,0);

\node [blue] at (4,-2.5) {Sinks};

\draw [densely dotted] (-1.2,.4) -- (-.7,.6);
\draw [densely dotted] (-1.2,-.4) -- (-.7,-.6);

\draw [densely dotted] (1.2,.4) -- (.7,.6);
\draw [densely dotted] (1.2,-.4) -- (.7,-.6);

\draw [->-=.5] (-1.7,.2) -- (-1.2,.4);
\draw [->-=.5] (-1.7,-.2) -- (-1.2,-.4);
\draw [->-=.5] (-1.7,-.6) -- (-1.2,-.4);
         
\draw [fill=black] (1.2,.4) circle (.05cm);
\draw [fill=black] (1.2,-.4) circle (.05cm);

\draw [fill=black] (0,0) circle (.05cm);
\draw [fill=black] (-1.2,.4) circle (.05cm);
\draw [fill=black] (-1.2,-.4) circle (.05cm);
\draw [fill=black] (-1.7,.6) circle (.05cm);
\draw [fill=black] (-1.7,.2) circle (.05cm);
\draw [fill=black] (-1.7,-.6) circle (.05cm);
\draw [fill=black] (-1.7,-.2) circle (.05cm);

\draw [fill=black] (1.7,.6) circle (.05cm);
\draw [fill=black] (1.7,.2) circle (.05cm);
\draw [fill=black] (1.7,-.6) circle (.05cm);
\draw [fill=black] (1.7,-.2) circle (.05cm);

\draw [->-=.5] (-1.7,.6) -- (-1.2,.4);
\draw [->-=.5] (-1.7,.2) -- (-1.2,.4);
\draw [->-=.5] (-1.7,-.2) -- (-1.2,-.4);
\draw [->-=.5] (-1.7,-.6) -- (-1.2,-.4);

\draw [->-=.6] (1.2,.4) -- (1.7,.6);
\draw [->-=.6] (1.2,.4) -- (1.7,.2);
\draw [->-=.6] (1.2,-.4) -- (1.7,-.2);
\draw [->-=.6] (1.2,-.4) -- (1.7,-.6);

\draw [densely dotted] (-2.5,0) -- (-2,0);
\draw [densely dotted] (2,0) -- (2.5,0);
\draw [densely dotted] (-1.2,-.2) -- (-1.2,.2);
\draw [densely dotted] (1.2,-.2) -- (1.2,.2);
\draw [densely dotted] (-1.7,.3) -- (-1.7,.5);
\draw [densely dotted] (-1.7,-.3) -- (-1.7,-.5);
\draw [densely dotted] (1.7,.3) -- (1.7,.5);
\draw [densely dotted] (1.7,-.3) -- (1.7,-.5);

\draw [->-=.5] (-1.2,.4) -- (0,0);
\draw [->-=.5] (-1.2,-.4) -- (0,0);

\draw [->-=.5] (0,0) -- (1.2,.4);
\draw [->-=.5] (0,0) -- (1.2,-.4);

\end{tikzpicture}
}
\caption{The directed graph representation of a wireline network. Nodes of the network are connected to each other via directed  links of limited capacity. The directed graph is allowed to have cycles. In this figure, there are $t$ source nodes and $u$ sink nodes. The models allows for shared secret keys between various subsets of the nodes. Each node produces its outputs on its ongoing links based on its inputs, shared keys and its own private randomness.}
\label{fig:Randomized_Network}
\end{figure}
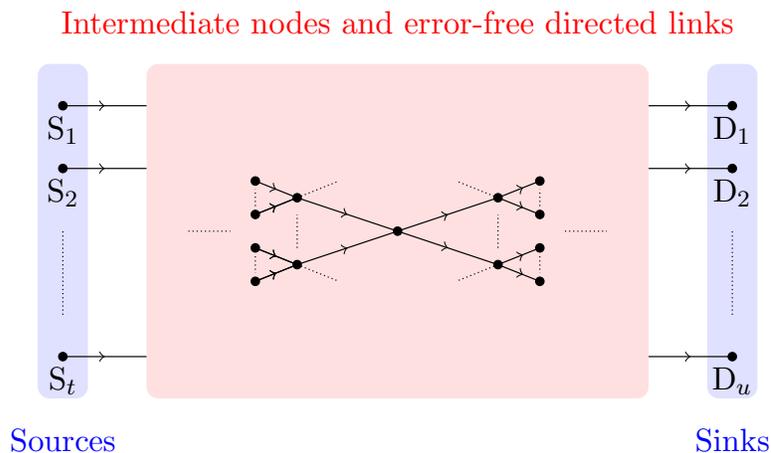

 There is also a set of keys $\mathbf{K}=\{K_1,K_2,\cdots,K_\Delta\}$ of limited rates, each of which is shared among a subset of the nodes. Hence, every node can use its available keys for encoding. Moreover, each  source or relay nodes can use a private randomness. Let us denote the set of all private randomness vectors by the set $\mathbf{W}=\{W_1,W_2,\cdots,W_\Theta\}$. Random variables $M_1,M_2,\cdots,M_t,K_1,K_2,\cdots,K_\Delta,W_1,W_2,\cdots,W_\Theta$ are mutually independent and uniform over their alphabet sets.

The edges of the wireline network have limited capacity. For a code of blocklength $n$, an edge with capacity $C_e$ can carry at most $n(C_e+\epsilon_n)$ bits where $\epsilon_n$ converges to zero as $n$ tends to infinity. Similarly, if the rate of message $M_i$ is $R_{M_i}$, then in a code of blocklength $n$, $M_i$ is a binary sequence of length $nR_{M_i}$. The same can be said of the rate of the shared keys $R_{K_i}$. The goal of the nodes of the network is to maximize the communication rates $R_{M_i}$ while minimizing the key rates $R_{K_i}$ as much as possible in such a way that the desired reliability (error probability condition at sinks) and security conditions are met.\footnote{Private randomness is commonly considered as a free resource and studying its rate is not of interest.} The resulting fundamental trade-off between $R_{M_i}$ and $R_{K_i}$ describes the capacity region of the problem.

Fixing a coding strategy by the nodes in the network, the eavesdropper will end up with a collection of observations from the network. We use the random variable $\mathbf{C}$ to denote all the information the eavesdropped has obtained. Random variable $\mathbf{C}$ is a function of $\mathbf{M}$, $\mathbf{K}$ and $\mathbf{W}$,
\begin{align*}
\mathbf{C}=f(\mathbf{M},\mathbf{K},\mathbf{W}).
\end{align*}

\underline{Linear Network Coding:}

In linear network coding, we assume that there is a finite field $\mathbb{F}$. Each variable $M_i$, $K_i$ and $W_i$ is a string of independent and uniformly distributed symbols from field $\mathbb{F}$. All the coding operations are restricted to taking weighted linear combinations in $\mathbb{F}$. Then, eavesdropper's information $\mathbf{C}$ can be expressed as
\begin{align}
\mathbf{C}={A}\mathbf{M}+{B}\mathbf{K}+{G}\mathbf{W},\label{eqn5}
\end{align}
for some matrices $A$, $B$ and $G$ where
\begin{align*}
\mathbf{M}&=\begin{bmatrix}
M_1,M_2,\cdots,M_t
\end{bmatrix}^{\mathrm{T}}\\
\mathbf{K}&=\begin{bmatrix}
K_1,K_2,\cdots,K_\Delta
\end{bmatrix}^{\mathrm{T}}\\
\mathbf{W}&=\begin{bmatrix}
W_1,W_2,\cdots,W_\Theta
\end{bmatrix}^{\mathrm{T}}.
\end{align*}

\underline{Decoding conditions:}

\begin{itemize}
\item[--] Zero-error decoding

Each receiver is able to decode its desired messages with exactly zero-error probability for every given blocklength.

\item[--] $\epsilon$-error decoding

Each receiver is able to recover its desired message with vanishing probability of error as the blocklength grows.
\end{itemize}

\underline{Secrecy conditions:}

\begin{itemize}
\item[--] Perfect Secrecy

Assuming that random variables $K_1,K_2,\cdots,K_\Delta,W_1,W_2,\cdots,W_\Theta$ are mutually independent and uniform over their alphabet sets, the conditional pmf
$p(\mathbf{C}=\mathbf{c}|\mathbf{M}=\mathbf{m})$ should not depend on the value of $\mathbf{m}$, for any given $\mathbf{c}$. Equivalently, for any distribution on input message set $\mathbf{M}$, we should have
\begin{equation}
\label{def:perfect_sec}
I(\mathbf{M};\mathbf{C})=0, \qquad \forall p_{\mathbf{M}}(\mathbf{m}),
\end{equation}
as long as the message set $\mathbf{M}$, the key set $\mathbf{K}$ and private randomness set $\mathbf{W}$ are mutually independent.

\item[--] Strong secrecy

In strong secrecy, the independence between $\mathbf{M}$ and $\mathbf{C}$ no longer exists. There are two definitions of $\epsilon$-strong secrecy in the literature \cite{yaar14}\cite[Lemma 1]{CsiszarNarayan}: given $\epsilon_1>0$, the first definition requires that
\begin{equation}
\label{def:strong_sec}
I(\mathbf{M};\mathbf{C})\leq\epsilon_1.
\end{equation}
The above equation can be also expressed in terms of KL divergence:
\begin{align*}
D(p_{\mathbf{M}\mathbf{C}}|| p_{\mathbf{M}}p_{\mathbf{C}})\leq\epsilon_1.
\end{align*}
The second definition of strong secrecy requires a bound on the total variation distance (instead of KL divergence). Given some $\epsilon_2>0$, we require
\begin{align}
\lVert p_{\mathbf{M}\mathbf{C}}- p_{\mathbf{M}}p_{\mathbf{C}}\rVert_1\leq\epsilon_2.
\end{align}
\begin{remark}\label{rmk:strong-2def}
(Connection between the two definitions). We claim that strong secrecy in terms of mutual information implies strong secrecy in terms of total variation distance, \emph{i.e.,} $\epsilon_1$ being small implies that $\epsilon_2$ is also small. The reverse is also true if one can show that strong secrecy in terms of total variation distance holds with an exponentially vanishing $\epsilon_2$.  To show this, let us denote the alphabet set of $\mathbf{M}$ by $\mathcal{M}$. It follows from \cite[Lemma 1]{CsiszarNarayan} that if $\epsilon_1$-strong secrecy of the first definition, and $\epsilon_2$-strong secrecy of the second definition hold, then
\begin{align*}
\frac{\log_2 e}{2}\epsilon_2^2\leq\epsilon_1\leq\epsilon_2\log{\frac{\lvert\mathcal{M}\rvert}{\epsilon_2}},
\end{align*}
provided that $\lvert\mathcal{M}\rvert>4$. Hence, if $\epsilon_1$ becomes small, $\epsilon_2$ also becomes small. For the reverse direction, assume that message $M_i$ takes values in $\{1,2,\cdots,2^{nR_i}\}$ where $n$ is the blocklength and $R_i$ is the rate of the $i$-th message. Then $\log\lvert\mathcal{M}\rvert=n\sum_{i}R_i$. If we can ensure that the value of $\epsilon_2$ decreases \emph{exponentially fast} in blocklength $n$, then $n\epsilon_2$ converges to zero as $n$ becomes large, and  $\epsilon_2\log\left({\lvert\mathcal{M}\rvert}/{\epsilon_2}\right)$ will also converge to zero. This will imply that  $\epsilon_1$ vanishes as $n$ tends to infinity. 
\end{remark}

\item[--] Weak secrecy

Similar to strong secrecy, $\mathbf{M}$ and $\mathbf{C}$ are not independent, instead of \eqref{def:perfect_sec} and \eqref{def:strong_sec}, we say that $\epsilon$-weak secrecy holds if:
\begin{equation}
\label{def:weak_sec}
I(\mathbf{M};\mathbf{C})\leq\epsilon\cdot H(\mathbf{M}).
\end{equation}
\end{itemize}

It follows from the above definitions that perfect secrecy condition  \eqref{def:perfect_sec} is stronger than strong secrecy condition \eqref{def:strong_sec}, which in turn is stronger than weak secrecy constraint \eqref{def:weak_sec}.

\section{Main Results}
\label{sec:main_results}

\subsection{Results for linear codes}

\begin{theorem}[From strong secrecy to perfect secrecy for linear codes]
\label{theorem:strong_perfect}
Take an arbitrary linear code $\mathscr{C}$, with adversary observing
$$\mathbf{C}={A}\mathbf{M}+{B}\mathbf{K}+{G}\mathbf{W},$$
as defined in \eqref{eqn5}. If each of the strong secrecy constraints hold for some $\epsilon<1$, \emph{i.e.,} either of 
\begin{equation*}
I(\mathbf{M};\mathbf{C})\leq\epsilon,
\end{equation*}
or
\begin{align*}
\lVert p_{\mathbf{M}\mathbf{C}}- p_{\mathbf{M}}p_{\mathbf{C}}\rVert_1\leq\epsilon.
\end{align*}
hold for some $\epsilon<1/2$, then the code $\mathscr{C}$ is also perfect secure, \emph{i.e.,} $I(\mathbf{M};\mathbf{C})=0$.
\end{theorem}
\begin{proof}
Assume that $I(\mathbf{M};\mathbf{C})>0$ where $\mathbf{C}={A}\mathbf{M}+{B}\mathbf{K}+{G}\mathbf{W}.$ We will show that $I(\mathbf{M};\mathbf{C})\geq 1$ and $\lVert p_{\mathbf{M}\mathbf{C}}- p_{\mathbf{M}}p_{\mathbf{C}}\rVert_1\geq 1/2$. This will conclude the proof.

Assume that $\mathbf{C}$ is a column vector of size $k$. We claim that one can find a non-zero column vector $\mathbf{z}$ of size $k$ such that $\mathbf{z}^\dagger B=\mathbf{z}^\dagger G=\mathbf 0$ are the zero vector, but $\mathbf{z}^\dagger A\neq \mathbf 0$ where $\dagger$ is the transpose operator. If this is not the case, the equation $\mathbf{z}^\dagger[{B},{G}]=\mathbf{0}$ implies that $\mathbf{z}^\dagger[A,{B},{G}]=\mathbf{0}$, showing that the null space $[{B},{G}]^\dagger$ is the same as the null space of  $[A,{B},{G}]^\dagger$. Hence, the rank of the matrix $[{A},{B},{G}]$ is equal to the rank of $[{B},{G}]$. Thus, the image of  the matrix ${A}$ is a subset of the image of $[{B},{G}]$. Let us call the image of $[{B},{G}]$ by $\mathscr{I}$, which is a linear subspace of $\mathbb{F}^k$. Since elements of vectors $\mathbf{K}$ and $\mathbf{W}$ are independently and uniformly distributed over $\mathbb{F}$, ${B}\mathbf{K}+{G}\mathbf{W}$ will be \emph{uniformly distributed} over  $\mathscr{I}$. Just like Shannon's one-time-pad strategy, this will imply that $\mathbf{C}={A}\mathbf{M}+({B}\mathbf{K}+{G}\mathbf{W})$ will be independent of ${A}\mathbf{M}$, and masked by ${B}\mathbf{K}+{G}\mathbf{W}$. To see this, note that for any value of $\mathbf{M}=\mathbf{m}$, we have ${A}\mathbf{m}\in \mathscr{I}$ and the vector 
 $\mathbf{C}={A}\mathbf{m}+{B}\mathbf{K}+{G}\mathbf{W}$ will be 
 uniformly distributed over  $\mathscr{I}$ as well. This is because $\mathscr{I}={A}\mathbf{m}+\mathscr{I}$ since $\mathscr{I}$ is a linear subspace. As a result, the conditional distribution $p(\mathbf{C}|\mathbf{m})$ does not depend on the value of $\mathbf{m}$. Hence, perfect secrecy condition holds. But this contradicts our assumption that $I(\mathbf{M};\mathbf{C})>0$. Thus, we can conclude that there is a non-zero column vector $\mathbf{z}$ of size $k$ such that $\mathbf{z}^\dagger B=\mathbf{z}^\dagger G=\mathbf 0$ are the zero vector, but $\mathbf{z}^\dagger A\neq \mathbf 0$. This implies that  $\mathbf{z}^\dagger C=\mathbf{z}^\dagger A\mathbf{M}\neq \mathbf 0$.

Now, observe that
\begin{align*}I(\mathbf{M};\mathbf{C})&\geq I(\mathbf{M};\mathbf{z}^\dagger\mathbf{C})
\\&=I(\mathbf{M};\mathbf{z}^\dagger A\mathbf{M})
\\&=H(\mathbf{z}^\dagger A\mathbf{M})
\\&\overset{(a)}{=}\log|\mathbb{F}|
\\&\geq 1,\end{align*}
where in $(a)$, we used the fact that $\mathbf{M}$ has uniform distribution, and hence $(\mathbf{z}^\dagger A)\mathbf{M}$ is a uniformly distributed symbol in $\mathbb{F}$. 

Next, defining functions $\hat{m}=f(\mathbf{m})=\mathbf{z}^\dagger A\mathbf{m}$ and $\hat{c}=g(\mathbf{c})=\mathbf{z}^\dagger \mathbf{c}=f(\mathbf{m})$, observe that $\hat{M}=\hat{C}$ is a uniform symbol in $\mathbb{F}$. Then, we can write
\begin{align*}
\lVert p_{\mathbf{M},\mathbf{C}}-p_{\mathbf{M}}\cdot p_{\mathbf{C}}\rVert_1
&\overset{(a)}{\geq} \lVert p_{\hat{M},\hat{C}}-p_{\hat{M}}\cdot p_{\hat{C}}\rVert_1\\
&=\frac 12\sum_{a,b\in\mathbb{F}}{\lvert\mathbbm{1}[a=b]\times\frac{1}{|\mathbb{F}|}-\frac{1}{|\mathbb{F}|^2}\rvert}
\\&=\frac{|\mathbb{F}|(|\mathbb{F}|-1)}{|\mathbb{F}|^2}=\left(1-\frac{1}{|\mathbb{F}|}\right)
\\&\geq \frac 12.
\end{align*}
where $\mathbbm{1}[\cdot]$  is the indicator function, and step $(a)$ follows from the data processing property of total variation distance (see e.g. \cite{PardoVa97}), which states that for any channel $p(y|x)$ we have
$$\|p(x)-q(x)\|_1\geq \|p(y)-q(y)\|_1$$
where $p(y)=\sum_{x}p(x)p(y|x)$ and $q(y)=\sum_x q(x)p(y|x)$. We get our desired inequality if we set the alphabet $\mathcal X$ to be the alphabet of $(\mathbf{M},\mathbf{C})$, $p(x)=p(\mathbf{m},\mathbf{c})$, $q(x)=p(\mathbf{m})p(\mathbf{c})$, and $p(y|x)$ to be the application of  functions $f$ and $g$ applied on the  $\mathbf{M}$ and $\mathbf{C}$ parts of $X$, respectively.

\end{proof}
\begin{theorem}[From $\epsilon$-error to zero-error for linear codes]
\label{theorem:epsilon_zero_error}
Take an arbitrary linear code $\mathscr{C}$ over a finite field $\mathbb{F}$. If the average error probability of a sink node is less than $1-1/|\mathbb{F}|$, then the error probability of the sink node has to be zero.
\end{theorem}
\begin{proof} Consider a sink node. The sink node receives a vector $\mathbf{Y}$ which is a linear combination of messages, keys and private randomness symbols. In other words, we have
\begin{equation*}
\mathbf{Y}={A}\mathbf{M}+{B}\mathbf{K}+{G}\mathbf{W}.
\end{equation*}
for some matrices $A,B$ and $G$. The message vector $\mathbf{M}$ can be split into two part $(\mathbf{M}_1, \mathbf{M}_2)$ where $\mathbf{M}_1$ is the set of messages that the sink nodes wants to decode, and $\mathbf{M}_2$ is the collection of other messages. Similarly, $\mathbf{K}$ can be split into two part $(\mathbf{K}_1, \mathbf{K}_2)$ where $\mathbf{K}_1$ is the set of secret keys that the sink nodes has, and $\mathbf{K}_2$ is the set of secret keys that are not shared with the sink node. Then, we can write
\begin{equation*}
\mathbf{Y}={A_1}\mathbf{M}_1+{A_2}\mathbf{M}_2+{B_1}\mathbf{K}_1+{B_2}\mathbf{K}_2+{G}\mathbf{W}.
\end{equation*}
Since the sink has vector $\mathbf{Y}$ and key $\mathbf{K}_1$, its task is to recover $\mathbf{M}_1$ from 
\begin{equation*}
\mathbf{Y}-{B_1}\mathbf{K}_1={A_1}\mathbf{M}_1+{A_2}\mathbf{M}_2+{B_2}\mathbf{K}_2+{G}\mathbf{W}.
\end{equation*}
Note that the sink node does not know any of $\mathbf{M}_2$, $\mathbf{K}_2$ or $\mathbf{W}$. These three variables $\mathbf{M}_2$, $\mathbf{K}_2$ or $\mathbf{W}$ are mutually independent and uniform over their alphabet sets. 
Let $\mathbf{Z}=\mathbf{Y}-{B_1}\mathbf{K}_1$. Given a value for  $\mathbf{Z}=\mathbf{z}$ for some  $\mathbf{m}_1$, we say that  $(\mathbf{z},\mathbf{m}_1)$ is a compatible pair if the equation
\begin{align}{A_2}\mathbf{m}_2+{B_2}\mathbf{k}_2+{G}\mathbf{w}=\mathbf{z}-{A_1}\mathbf{m}_1\label{eqn:AM20}\end{align}
has a solution in variables $\mathbf{m}_2, \mathbf{k}_2, \mathbf{w}$.

Given a pair $(\mathbf{z},\mathbf{m}_1)$, two possibilities might occur
\begin{itemize}
\item The pair $(\mathbf{z},\mathbf{m}_1)$ are not compatible. In this case, $p(\mathbf{m}_1|\mathbf{z})=0$ and the sink is certain that its intended message is not equal to $\mathbf{m}_1$. 
\item The pair $(\mathbf{z},\mathbf{m}_1)$ are compatible, and the equation 
\begin{align}{A_2}\mathbf{m}_2+{B_2}\mathbf{k}_2+{G}\mathbf{w}=\mathbf{z}-{A_1}\mathbf{m}_1\label{eqn:AM1}\end{align}
has at least one solution for $\mathbf{m}_2, \mathbf{k}_2, \mathbf{w}$. Then, note that
 the number of solutions $(\mathbf{m}_2, \mathbf{k}_2, \mathbf{w})$ that satisfy \eqref{eqn:AM1} is fixed and determined by the dimension of the null space of matrix $[A_2,B_2, G]$. Since $\mathbf{M}_2$, $\mathbf{K}_2$ and $\mathbf{W}$ are mutually independent and uniform, $p(\mathbf{m}_1|\mathbf{z})$ is equal to the number of solutions $(\mathbf{m}_2, \mathbf{k}_2, \mathbf{w})$ of \eqref{eqn:AM1}, divided by the total number of triples $(\mathbf{m}_2, \mathbf{k}_2, \mathbf{w})$. This implies that from the perspective of the sink that has vector $\mathbf{z}$, all the messages $\mathbf{m}_1$ that are compatible with $\mathbf{z}$ are equally likely to have been the transmitted message. 
\end{itemize}
Assume that the sink's error probability is positive. We show that for \emph{any} vector $\mathbf{z}$ that the sink may end up with, there are at least $|\mathbb{F}|$ sequences $\mathbf{m}_1$ that are compatible with $\mathbf{z}$. Thus, the chance of correct decoding will be at most $1/|\mathbb{F}|$. This would complete the proof. 

Now, if the sink's error probability is positive, there exists some vector $\mathbf{z}$
and two distinct compatible sequences $\mathbf{m}'_1\neq \mathbf{m}^*_1$ with it, \emph{i.e.,}  the following two equations have solutions  $(\mathbf{m}_2, \mathbf{k}_2, \mathbf{w})$ and  $(\mathbf{m}_2', \mathbf{k}_2', \mathbf{w}')$:
\begin{align}{A_2}\mathbf{m}^*_2+{B_2}\mathbf{k}^*_2+{G}\mathbf{w}^*&=\mathbf{z}-{A_1}\mathbf{m}^*_1\label{eqn:AM32}
\\
{A_2}\mathbf{m}'_2+{B_2}\mathbf{k}'_2+{G}\mathbf{w}'&=\mathbf{z}-{A_1}\mathbf{m}'_1
\label{eqn:AM3}\end{align}
By subtracting these two equations, we get that for $\mathbf{m}''_1=\mathbf{m}^*_1-\mathbf{m}'_1\neq \mathbf{0}$, the equation
\begin{align}{A_2}\mathbf{m}''_2+{B_2}\mathbf{k}''_2+{G}\mathbf{w}''&=-{A_1}\mathbf{m}''_1\label{eqn:AM243}\end{align}
has a solution  $(\mathbf{m}''_2, \mathbf{k}''_2, \mathbf{w}'')= (\mathbf{m}^*_2, \mathbf{k}^*_2, \mathbf{w}^*)- (\mathbf{m}'_2, \mathbf{k}'_2, \mathbf{w}')$. 

Now take \emph{any} vector $\mathbf{z}$ that the sink may end up with, and let $\mathbf{m}_1$ be the true message sequence that is compatible with $\mathbf{z}$. We claim that $\mathbf{z}$ is also compatible with $\mathbf{m}_1+\alpha \mathbf{m}''_1$ for any $\alpha\in\mathbb{F}$. This follows from multiplying both sides of \eqref{eqn:AM243} by $\alpha$ and then adding it up with \eqref{eqn:AM20}. Since $\mathbf{m}''_1\neq 0$, the sequences 
$\mathbf{m}_1+\alpha \mathbf{m}''_1$ for different values of $\alpha$ are distinct vectors. Since $\alpha$ has $|\mathbb{F}|$ possibilities, this shows that there are at least $|\mathbb{F}|$ sequences $\mathbf{m}_1$ that are compatible with $\mathbf{z}$. 
\end{proof}

\subsection{Result for linear and non-linear codes}
Given message rates $R_{M_i}$, $i=1,2,\dots, t$ and key rates $R_{K_i}$ for $i=1,2,\dots, \Delta$, we say that these message and key rates are asymptotically weakly secure achievable if  there is a sequence of codes $\mathscr{C}_j$ whose message and key rates converge to $R_{M_i}$, $i=1,2,\dots, t$ and $R_{K_i}$ for $i=1,2,\dots, \Delta$ as $j$ tends to infinity, and furthermore, $\mathscr{C}_j$  is $\epsilon_j$-weakly secure, \emph{i.e.,} satisfying
\begin{equation*}
I(\mathbf{M};\mathbf{C})\leq\epsilon_j H(\mathbf{M}),
\end{equation*}
for some vanishing sequence  $\epsilon_j\rightarrow 0$ as $j$ tends to infinity. We say that the given message and key rates are asymptotically weakly secure achievable with linear codes if one can find a  sequence  of linear codes $\mathscr{C}_j$ with the above properties.

We say that message rates $R_{M_i}$, $i=1,2,\dots, t$ and key rates $R_{K_i}$ for $i=1,2,\dots, \Delta$, are asymptotically strongly secure achievable if a similar condition holds except that we require $\mathscr{C}_j$ to be $\epsilon_j$-strongly secure 
\begin{equation*}
I(\mathbf{M};\mathbf{C})\leq\epsilon_j,
\end{equation*} 
for some vanishing sequence  $\epsilon_j$. Asymptotically strongly secure achievable rates with linear codes are defined similarly.

\begin{theorem}[From weak secrecy to strong secrecy for linear and non-linear codes]
\label{theorem:weak_strong}
Any message and key rates $R_{M_i}$ and $R_{K_i}$ that is asymptotically weakly secure achievable, is also asymptotically strongly secure achievable. Also, any message and key rates $R_{M_i}$ and $R_{K_i}$ that is asymptotically weakly secure achievable with linear codes is also asymptotically strongly secure achievable with linear codes.
\end{theorem}

In order to prove the above theorem, we need tools from random binning of sources that are given in Appendix \ref{appendixA}.


\begin{proof}[Proof of Theorem \ref{theorem:weak_strong}] We begin by providing the high level structure of the proof.

\textbf{High level structure of the proof:}
Suppose we have a code $\mathscr{C}$ satisfying the weak secrecy condition with parameter $\epsilon_a$, \emph{i.e.,} \begin{align}I(\mathbf{M};\mathbf{C})\leq\epsilon_a\cdot H(\mathbf{M}). \label{eqn:weaksecrecy}\end{align} Also assume that the error probability of the code is $\epsilon_b$. Then, we construct a  sequence of strongly-secure codes $\mathscr{C}'_n$ whose information leakage vanishes as $n$ tends to infinity. The message rates of $\mathscr{C}'_n$ converge to a number that is at least $R_{M_i}-\beta$, and the key rates of $\mathscr{C}'_n$ converge to a number that is at most $R_{K_i}+\beta$. Here $\beta$ is a constant that depends only on $\epsilon_a$ and $\epsilon_b$. Furthermore, $\beta$ converges to zero as $\epsilon_a$ and $\epsilon_b$ converge to zero. Constructing this sequence of strongly secure codes completes the proof. This sequence of codes is constructed by repeating the original code $\mathscr{C}$ and properly appending the repeated code.

\textbf{Some definitions:}
Assume that there are $u$ sink nodes and message $M_i$ is desired by sinks $\mathcal{T}_i\subseteq [u]$. Let us denote by $\hat{{M}}_{ij}$ to be the reconstruction of $M_i$ by sink $j\in\mathcal{T}_i$. Since the error probability of the code $\mathscr{C}$ is $\epsilon_b$, By Fano's inequality, we have \begin{align}H(M_i|\hat{{M}}_{ij})\leq h(\epsilon_b)+\epsilon_b\log|{\mathcal{M}_i}|, \qquad \forall j\in \mathcal{T}_i. \label{eqndefdeltai}\end{align} Let $\delta_i=h(\epsilon_b)+\epsilon_b\log|{\mathcal{M}_i}|$, and \begin{align}\delta=\max_{i\in[t]}\delta_i.\label{defdelta}\end{align}
If we fix the coding operations at all nodes, the output reconstructions and  eavesdropper's information will be functions of  the message $\mathbf{M}$, secret key $\mathbf{K}$ and private randomness $ \mathbf{W}$: 
$$(\mathbf{\hat{M}}, \mathbf{C})=g(\mathbf{M}, \mathbf{K},  \mathbf{W}).$$

\textbf{Independent repetitions of the code $\mathscr{C}$:}
Assume that we independently run the above code $n$ times. In other words, instead of considering one copy of message $M_i$, assume that $n$ i.i.d.~copies $M_{i}(1), M_{i}(2), \cdots, M_{i}(n)$ exist for $i\in [t]$.  For each of the $n$ copies of the messages, we run the given code and the sinks produce reconstructions $\hat{{M}}_{ij}(1), \hat{{M}}_{ij}(2), \cdots, \hat{{M}}_{ij}(n)$ for $i\in[t], j\in\mathcal{T}_i$. We call this expansion $n$ i.i.d.~repetitions of the code and denote it by $\mathscr{C}^n$. Observe that the rate of the expanded code $\mathscr{C}^n$ is equal to the rate of the original code $\mathscr{C}$, because even though the links in the network are used $n$ times a single code, but the message communicated over the network is also multiplied by $n$. Similarly, the rates of secret keys shared among the network nodes remain unchanged. By summing up the weak secrecy conditions $I(\mathbf{M}(i);\mathbf{C}(i))\leq\epsilon_a\cdot H(\mathbf{M}(i))$ for each repetition of the code, we obtain that 
$$I(\mathbf{M}([n]);\mathbf{C}([n]))\leq\epsilon_a\cdot H(\mathbf{M}([n])),$$
where $\mathbf{M}([n])=\{\mathbf{M}(1), \mathbf{M}(2), \dots, \mathbf{M}(n)\}$
is the collection of all messages of $\mathscr{C}^n$. We see that the weak secrecy condition holds with the same parameter $\epsilon_a$ for $\mathscr{C}^n$. However, the error probability of the expanded code $\mathscr{C}^n$ is higher, because $\mathscr{C}^n$ will be in error if an error occurs in \emph{any} of the $n$ iterations of the code. Nonetheless, by properly appending the expanded space provided by $\mathscr{C}^n$, we not only bring down the error probability, but also go  from weak secrecy to strong secrecy at the cost of sacrificing an asymptotically vanishing reduction in message rates.

We can represent the expanded code $\mathscr{C}^n$ by i.i.d.\ variables $(\mathbf{\hat{M}}(i),\mathbf{C}(i),\mathbf{M}(i),\mathbf{K}(i), \mathbf{W}(i))$ for $i\in [n]$, and follows that 
$$(\mathbf{\hat{M}}(i), \mathbf{C}(i))=g(\mathbf{M}(i), \mathbf{K}(i),  \mathbf{W}(i)).$$

\textbf{Informal sketch of the proof:}
Since the formal proof involves several technical details that might clutter the flow of ideas, 
we begin by the informal sketch of the proof to convey the essential ideas.  The formal proof is given afterwards. Below, we use the term ``small" informally to mainly denote a term that vanishes as $\epsilon_a$ and $\epsilon_b$ converge to zero. 

Via a binning argument, we find two appropriate functions of $M_{i}([n])$, namely $\widetilde{M}_i$ and $F_i$ for $i\in [t]$, such that 
\begin{itemize}
\item[i)] The alphabet size of variable $F_i$ is small for any $i\in[t]$. 
\item[ii)] Random variable $\widetilde{M}_i$ is almost uniformly distributed. Since $\widetilde{M}_i$ is a function of $M_{i}([n])$, multiple sequences $m_{i}([n])$ may be mapped to the same $\widetilde{m}_i$. We construct the function $\widetilde{M}_i$ such that the number of $M_{i}([n])$ that are mapped to each realization of $\widetilde{M}_i$ is small. Hence $\widetilde{M}_i$ is in an approximate one-to-one map with $M_{i}([n])$, and the entropy of random variable $\widetilde{M}_i$ is almost equal to the entropy of $M_{i}([n])$ for any $i\in[t]$. 
\item[iii)] Let us use $\widetilde{\mathbf{M}}$ and ${\mathbf{F}}$ to denote the collection of $\widetilde{M}_i$ and $F_i$ for $i\in [t]$, respectively. Let $\widetilde{\mathbf{C}}=(\mathbf{C}([n]), \mathbf{F})$. Then, $\widetilde{\mathbf{M}}$ and $\widetilde{\mathbf{C}}$ are almost mutually independent. In particular, there is some $\eta>0$ such that
\begin{align}
\lVert p_{\widetilde{\mathbf{M}}\widetilde{\mathbf{C}}}- p_{\widetilde{\mathbf{M}}}p_{\widetilde{\mathbf{C}}}\rVert_1\leq 2^{-\eta n}.
\end{align}
\item[iv)] Given $j\in \mathcal{T}_i$, as mentioned above, error probability $\mathbb{P}(M_i([n])\neq \hat{M}_{ij}([n]))$ can become large when $n$ becomes large; this is because the expanded code will be in error if an error occurs in \emph{any} of the $n$ iterations of the code. However, variable  $F_i$ is such that one can recover $M_i([n])$ from $F_i$ and reconstruction $\hat{M}_{ij}([n])$ for any $j\in \mathcal{T}_i$ with very high probability. In other words, once given $F_i$, it is possible to use the $n$ noisy reconstructions $\hat{M}_{ij}([n])$ to recover $M_i([n])$ with high probability. Thus, providing the additional variable $F_i$ to the receivers will be used to resolve the error probability issue.
\end{itemize}
Now, we show that how finding  $\widetilde{M}_i$ and $F_i$ with the above properties completes the proof. Since $M_{i}([n])$ are mutually independent for $i\in[t]$, we have that $\widetilde{M}_i$'s are also mutually independent for $i\in[t]$. We view $\widetilde{M}_i$ as the messages for the new code $\widetilde{\mathscr{C}}$ that we construct. Thus, each source node that was receiving message $M_i$, is now receiving $\widetilde{M}_i$ as the $i$-th message. But to be able to exploit the original expanded code $\mathscr{C}^n$, we need to create $M_i([n])$ from $\widetilde{M}_i$. To do this, we consider the channel $p_{M_i([n])|\widetilde{M}_i}$, and pass $\widetilde{M}_i$ through this channel  to simulate $M_i([n])$. Since $M_i([n])$ is uniformly distributed, this simulation is nothing but looking at sequences $M_i([n])$ that are mapped to the same $\widetilde{M}_i$, and choosing uniformly at random from them. 
Since the $i$-th message $\widetilde{M}_i$ may be available at multiple source nodes, we should make sure that they all create the same  $M_i([n])$. To do this,  we assume an additional common key is shared among the source nodes to make this coordination. The rate of this extra key will be shown to vanish in the limit by using property (ii). The source nodes can use this shared key to simulate the channel $p_{M_i([n])|\widetilde{M}_i}$, and pass $\widetilde{M}_i$ to obtain one common copy of $M_i([n])$. Having simulated $M_i([n])$, the nodes can find $F_i$ (which is a function of $M_i([n])$). 

Once $M_i([n])$'s are simulated, we can use the encoding and decoding operations of $\mathscr{C}^n$. This allows the sinks to produce reconstructions $\hat{{M}}_{ij}(1), \hat{{M}}_{ij}(2), \cdots, \hat{{M}}_{ij}(n)$. 
Next, $F_i$'s are also sent from source nodes to sink nodes via the network links. Since the entropy rates of $F_i$'s are vanishing, we do not violate the link capacities asymptotically. From property (iv) given above, this will allow sinks to decode their intended messages with vanishing error probability. From property (iii) given above, strong secrecy condition (total variation distance definition) holds even if eavesdropper also gets to learn $F_i$'s in addition to $\mathbf{C}([n])$. Since the total variation distance drops exponentially fast in $n$, from Remark \ref{rmk:strong-2def}, we get strong secrecy condition in the sense of vanishing mutual information. This will complete the proof.

\textbf{Formal proof:}

\textbf{Step 1: Construction of $\widetilde{M}_i$ and $F_i$ for $i\in [t]$:}

Let $R_i=\log|\mathcal M_i|$. This quantity is proportional to $R_{M_i}$ of code $\mathscr{C}$. In fact, if code  $\mathscr{C}$ consists of $k$ uses of the network, then $R_{M_i}=R_i/k$ is the message sent per network use. 
 Let
\begin{align}\tilde{R}_i&=R_i-2\epsilon_a\cdot H(\mathbf{M})-2\delta\label{eqnRtudfi},
\\R_{F_i}&=2\delta\label{eqnRFi},
\end{align}
where $\delta$ was defined in \eqref{defdelta}. 

Observe that the repetitions of message $M_i$, \emph{i.e.,} $M_{i}([n])$ has alphabet set $\mathcal{M}_i^n$. We consider two independent binnings of $\mathcal{M}_i^n$, one  into $2^{n\tilde{R}_i}$ bins and another into $2^{nR_{F_i}}$ bins. These binnings are done randomly and independently. Applying the (random) binning mapping to $M_{i}([n])$, let us denote the bin indices by $\widetilde{M}_i$ and $F_i$, respectively.  The binning mappings can be linear or non-linear depending on whether we are proving the theorem for linear or non-linear case. 

According to Theorem \ref{thm:errorvanishfast} given in the appendix, if for any $\mathcal{S}\subseteq [t]$, the binning rate vector $$(\tilde{R}_1, {R}_{F_1},\tilde{R}_2, {R}_{F_2},\cdots,\tilde{R}_t, {R}_{F_t})$$ satisfies the following inequality,
\begin{align}
\label{OSRB_Rate_Condition}
\sum_{i\in\mathcal{S}}{\tilde{R}_i+R_{F_i}}< H({M}_{\mathcal{S}}|\mathbf{C})&=H({M}_{\mathcal{S}})-I({M}_{\mathcal{S}};\mathbf{C})\nonumber\\
&=\sum_{i\in\mathcal{S}}{R_i}-I({M}_{\mathcal{S}};\mathbf{C}),
\end{align}
then, one can find $\kappa>0$ such that  for sufficiently large enough $n$
\begin{equation}
\label{eq:total_var_strong_sec}
\mathbb{E}\lVert P_{\widetilde{\mathbf{M}}{\mathbf{F}}\mathbf{C}([n])}-p_{\widetilde{\mathbf{M}}}^Up_{{\mathbf{F}}}^U p_{\mathbf{C}([n])}\rVert_1\leq 2^{-\kappa n}
\end{equation}
where the expected value is over all random binning mappings and $p^U$ is the uniform distribution. Observe that  \eqref{OSRB_Rate_Condition} holds by the choice of $\tilde R_i$ and $R_{F_i}$ given in \eqref{eqnRtudfi} and \eqref{eqnRFi}. The reason is that
\begin{align*}
\sum_{i\in\mathcal{S}}{\tilde{R}_i+R_{F_i}}&=\left(\sum_{i\in\mathcal{S}}R_i\right)-2\epsilon_a|\mathcal S|\cdot H(\mathbf{M})\\
&\overset{(a)}{\leq} \left(\sum_{i\in\mathcal{S}}R_i\right)-\epsilon_a|\mathcal S|\cdot H(\mathbf{M})-|\mathcal S|\cdot  I(\mathbf{M};\mathbf{C})\\
&\leq \left(\sum_{i\in\mathcal{S}}R_i\right)-\epsilon_a|\mathcal S|\cdot H(\mathbf{M})-  I({M}_{\mathcal{S}};\mathbf{C})\\
&< \left(\sum_{i\in\mathcal{S}}R_i\right)-  I({M}_{\mathcal{S}};\mathbf{C}) \label{eqn:tofollowfromf},
\end{align*}
where $(a)$ follows from \eqref{eqn:weaksecrecy}.

Next, we want to define some Slepian-Wolf decoders. Csisz\'ar in \cite[Theorem 1,3]{csiszar82} proves the existence of error exponents for the the Slepian-Wolf theorem \cite{SlepWo73} for random non-linear and linear binning. This result implies that we can recover $M_i([n])$ from bin index $F_i$ and side information $\hat{M}_{ij}([n])$ for any $j\in \mathcal{T}_i$ with error probability of at most $2^{-n\beta_i}$ for some $\beta_i>0$  if 
$$R_{F_i}>H(M_i|\hat{M}_{ij}),$$
and $n$ is sufficiently large.
Note that the probability of success of the Slepian-Wolf decoder is with respect to random binning (computed by taking the statistical average over all random binnings). Observe that $R_{F_i}$ given in \eqref{eqnRFi} satisfies this inequality because of \eqref{eqndefdeltai} and \eqref{defdelta}.

Let 
\begin{align}
R_{G_i}&=2\epsilon_a\cdot H(\mathbf{M})+3\delta.
\end{align}
Because $R_{G_i}+\tilde{R}_i>H(M_i)$, by Theorem \ref{thm:thm7}, one can simulate the channel
$p_{M_i([n])|\widetilde{M}_i}$ using randomness of rate $R_{G_i}$  within an average total variation distance of at most $2^{-n\zeta_i}$ for some $\zeta_i>0$.

We claim that there is a \emph{deterministic} binning such that for some $\eta>0$, 
\begin{itemize}
\item (i) We have
\begin{equation}
\label{eq:total_var_strong_sec22n2}
\lVert p_{\widetilde{\mathbf{M}}{\mathbf{F}}\mathbf{C}([n])}-p_{\widetilde{\mathbf{M}}}^Up_{{\mathbf{F}}}^U p_{\mathbf{C}([n])}\rVert_1\leq 2^{-\eta n}.
\end{equation}
\item (ii) For any $i$, with probability $1-2^{-\eta n}$, one can recover  $M_i([n])$ from bin index $F_i$ and side information $\hat{M}_{ij}([n])$ for any $j\in \mathcal{T}_i$.
\item (iii) For any $i$, one can simulate the channel
$p_{M_i([n])|\widetilde{M}_i}$ using randomness of rate $R_{G_i}$  within a total variation distance of at most $2^{-n\eta}$.
\end{itemize}
 The reason is that we know the average of the sum of the  total variation distance of \eqref{eq:total_var_strong_sec22n2}, plus the error probabilities of the Slepian-Wolf decoders, plus the total variation distance of the channel simulator converges to zero (exponentially fast) over all random instances. Hence, there must exist a deterministic binning (a fixing of binnings) that makes this total sum converge to zero (exponentially fast). 

\textbf{Step 2: Completing the proof using $\widetilde{M}_i$ and $F_i$ for $i\in [t]$:}

We construct a new code $\widetilde{\mathscr{C}}$ as follows: the $i$-th message is denoted by $\widetilde{M}_i$ and is uniformly distributed over a set of size $2^{n\tilde{R}_i}$. The nodes of the network also have shared keys of the same length as they have in $\mathscr{C}^n$. Additionally, the source nodes who obtain the $i$-th message $\widetilde{M}_i$, are assumed to share a common secret key of rate $R_{G_i}$. This secret key is used by them to simulate the same channel $p_{M_i([n])|\widetilde{M}_i}$. The source nodes pass their messages  $\widetilde{M}_i$ through this channel to produce $M_i([n])$. Having produced $M_i([n])$, the nodes can find $F_i$ (which is a function of $M_i([n])$). Furthermore, with their simulated $M_i([n])$, we can use the encoding and decoding operations of $\mathscr{C}^n$. This gives the adversary random variable $\mathbf{C}([n])$. Furthermore, the source nodes send variables $F_i$ through the network links. This comes at a negligible additional cost since $R_{F_i}$ can be made arbitrarily small. This gives the adversary random variables $\mathbf{C}([n])$ and $\mathbf{F}$.

\emph{Secrecy and reliability analysis:} Observe that the induced pmf on  $\widetilde{M}_i, M_i([n])$ and $F_i$ is as follows:
$$p_{\widetilde{\mathbf{M}}}^U~\cdot~\tilde{p}_{\mathbf{M}([n])~\mathbf{|}~\widetilde{\mathbf{M}}}~\cdot~p_{\mathbf{F},\mathbf{C}([n])~|~\mathbf{M}([n])}$$
Since by \eqref{eq:total_var_strong_sec22n2}, $$\|p_{\widetilde{\mathbf{M}}}^U-p_{\widetilde{\mathbf{M}}}\|_1\leq 2^{-\eta n}$$ and by (iii), $$\|p_{\widetilde{\mathbf{M}}}\tilde{p}_{\mathbf{M}([n])|\widetilde{\mathbf{M}}}-p_{\widetilde{\mathbf{M}}}p_{\mathbf{M}([n])|\widetilde{\mathbf{M}}}\|_1\leq  2^{-\eta n}$$
using \cite[Lemma 3, part 3]{yaar14}, we get that
\begin{align}\|p_{\widetilde{\mathbf{M}}}^U~\cdot~\tilde{p}_{\mathbf{M}([n])|\widetilde{\mathbf{M}}}~\cdot~p_{\mathbf{F},\mathbf{C}([n])|\mathbf{M}([n])}~-~ p_{\widetilde{\mathbf{M}}}~\cdot~p_{\mathbf{M}([n])|\widetilde{\mathbf{M}}}~\cdot~p_{\mathbf{F},\mathbf{C}([n])|\mathbf{M}([n])}\|_1\leq 2\times 2^{-\eta n}.\label{eqnLkhgdlk}\end{align}
Hence, the induced pmf of the  code $\widetilde{\mathscr{C}}$ is very close to the induced pmf of  $\mathscr{C}^n$ with $\widetilde{M}_i$ and $F_i$ created as deterministic bin indices of $M_i([n])$. 
From \eqref{eq:total_var_strong_sec22n2}, we can then conclude that in the new code $\widetilde{\mathscr{C}}$, the message vector $\widetilde{\mathbf{M}}$ is almost independent of $\mathbf{F},\mathbf{C}([n])$. Since the
strong secrecy condition (total variation distance definition) holds with the total variation distance dropping exponentially fast in $n$, from Remark \ref{rmk:strong-2def}, we get strong secrecy condition in the sense of vanishing mutual information between  $\widetilde{\mathbf{M}}$ and $\mathbf{F},\mathbf{C}([n])$.

The sink nodes use the encoding and decoding operations of $\mathscr{C}^n$. This allows the sinks to produce reconstructions $\hat{{M}}_{ij}(1), \hat{{M}}_{ij}(2), \cdots, \hat{{M}}_{ij}(n)$. Since  $F_i$'s are also sent from source nodes to sink nodes via the network links, from property (ii) given above, the sinks can decode their intended messages with vanishing error probability. This  completes the proof.

\end{proof}

\section{Conclusion}
\label{Conclusion_Sec}

In this paper, we considered a setup which contains $t$ transmitter, $u$ receivers and some intermediate nodes being connected with directed error-free point-to-point links. It is also assumed that there exists an eavesdropper being able to hear a certain subset of links. In order to provide secrecy, each node has access to some keys and private randomness. Defining different conditions on decoding error and secrecy, \emph{i.e.,} zero and $\epsilon$-error decoding; and weak, strong and perfect secrecy constraints, we were seeking to find a relation between rate regions considering different conditions. In Theorem \ref{theorem:strong_perfect}, we showed that for the linear case the rate region with strongly-secure condition is equivalent to one with perfectly-secure constraint. Theorem \ref{theorem:epsilon_zero_error} states the equivalency of $\epsilon$-error to zero-error rate region for the linear case. Moreover, we showed in Theorem \ref{theorem:weak_strong} for general case (both linear and non-linear regime) that relaxing the secrecy condition from strong to weak secrecy, does not change the rate region when we have an $\epsilon$-error decoding condition. Our conjecture is that the $\epsilon$-error weakly-secure rate region is equivalent to zero-error perfectly-secure one in the general case.

\section*{Acknowledgement}
The authors would like to thank Mohammad Hossein Yassaee for his helpful comments.


\bibliographystyle{IEEEtran}

\appendix
\section{Tools from random binning}\label{appendixA}

\subsection{Some Definitions}

\underline{Random binning:} In random binning, each realization of a random variable is randomly mapped to a bin index. Therefore, random binning is a random function like $\mathfrak{B}:\mathcal{M}\rightarrow\bar{\mathcal{M}}$ which uniformly and independently maps each symbol $m\in\mathcal{M}$ to a symbol $\bar{m}\in\bar{\mathcal{M}}$. In other words, $B=\mathfrak{B}(m)$ is a uniform random variable on the set $\{0,1,\cdots,\lvert \bar{\mathcal{M}}\rvert-1\}$ and for any $m_1\neq m_2\in\mathcal{M}$, $B_1=\mathfrak{B}(m_1)$ is independent of $B_2=\mathfrak{B}(m_2)$.

\underline{Linear random binning:} In linear random binning, the mapping function $\mathfrak{B}$ is linear. Each (affine) linear random binning has a matrix representation of the form $\bar{M}={A}M+V$, where $A$ is a random matrix, and $V$ is a random vector, all with independent and uniform entries in $\mathbb{F}$. Consider $M$ as a sequence of symbols in the finite field $\mathbb{F}$ with the length of $\ell_m$ and bin index $\bar{M}$ as a sequence of length $\ell_{\bar{m}}$ in $\mathbb{F}$, linear random binning matrix will be of size ${A}_{\ell_{\bar{m}}\times \ell_m}$ and $V$ will be of length $\ell_{\bar{m}}$. 

\underline{Distributed random binning:} 

In distributed random binning, there are a set of random functions $\mathfrak{B}_i:\mathcal{M}_i\rightarrow\bar{\mathcal{M}}_i,~i\in[t]$ where each $\mathfrak{B}_i$ is a random binning function and $\mathfrak{B}_i$'s are mutually independent. Distributed linear random binning can be characterized by matrices $A_i$ and drift terms $V_i$, $\bar{M}_i={A_i}M_i+V_i$ where entries of all of $A_i$ and $V_i$ are mutually independent and uniform over $\mathbb{F}$. Observe that the following facts holds in a distributed linear binning: (i) uniformity property: for any values of $m_i$ and $\bar{m}_i$, we have
\begin{align}
\mathbb{P}\left({A}_i m_i+V_i=\bar{m}_i\right)=\frac{1}{\lvert \bar{\mathcal{M}}_i\rvert},\label{uniformity}
\end{align}
and (ii) pairwise independence property: for any values of $m_i, m_j, \bar{m}_i$ and $\bar{m}_j$, we have
\begin{align}
\mathbb{P}({A}_i m_i+V_i=\bar{m}_i,{A}_j m_j+V_j=\bar{m}_j)=\frac{1}{\lvert \bar{\mathcal{M}}_i\rvert^2}.\label{uniformity-pair}
\end{align}

\subsection{Output Statistics of Random Binning}

Output Statistics of Random Binning (OSRB) is a tool introduced in \cite{yaar14} to describe the joint pmf of bin indices of multiple random variables.

\begin{theorem}[OSRB Theorem - Theorem~1 in \cite{yaar14}]\label{thm:OSRBa0}
Consider dependent random variables $(M_1,M_2,\cdots,M_t,C)$ with joint pmf $p(m_1,m_2,\cdots,m_t,c)$ on the finite alphabet set $\prod_{i=1}^{t}{\mathcal{M}_i}\times\mathcal{C}$. 
Let $\mathbf{M}^n, C^n$ be $n$ i.i.d.\ repetitions of $(\mathbf{M}, C)$ where $\mathbf{M}=(M_1, M_2, \cdots, M_t)$, \emph{i.e.,}
$$p(\mathbf{m}^n,c^n)=\prod_{i=1}^np(\mathbf{m}_i, c_i).$$
 Moreover, we assume that distributed random binning function $\mathfrak{B}_i:\mathcal{M}_i^n\rightarrow\bar{\mathcal{M}_i}=[2^{nR_i}],~i\in[t]$ maps each sequence of $\mathcal{M}_i^n$ independently and uniformly to the bin index set $[2^{nR_i}]$ that induces the following pmf
\begin{equation*}
P(\mathbf{m}^n,c^n,\bar{\mathbf{m}})=p(\mathbf{m}^n,c^n)\cdot\prod_{i=1}^{t}{\mathbbm{1}\left[\mathfrak{B}_i(m_i^n)=\bar{m}_i\right]}.
\end{equation*}
where $\mathbf{m}^n=\{m_i^n,~i\in[t]\}$ and $\bar{\mathbf{m}}=\{\bar{m}_i\in[2^{nR_i}],~i\in[t]\}$. Note that $P(\mathbf{m}^n,c^n,\bar{\mathbf{m}})$ shown by capital letter is a random pmf which is equal to $p_{\mathbf{m}^n,c^n|\mathfrak{B}_1,\mathfrak{B}_2,\cdots,\mathfrak{B}_t}$ for each fixed binning. According to the OSRB theorem, if for each $\mathcal{S}\subseteq [t]$, the binning rate vector $\left(R_1,R_2,\cdots,R_t\right)$ satisfies the inequality,
\begin{equation*}
\sum_{i\in\mathcal{S}}{R_i}< H(\mathbf{M}_{\mathcal{S}}|C),
\end{equation*}
the expected value of the total variation of the joint pmf $P(c^n,\bar{\mathbf{m}})$ from the $p_{c^n}\prod_{i=1}^{t}p_{[2^{nR_i}]}^U$ tends to zero as $n$ approaches infinity:
\begin{equation}
\lim_{n\rightarrow\infty}\mathbb{E}_{\mathfrak{B}}\lVert P(c^n,\bar{\mathbf{m}})-p_{c^n}\prod_{i=1}^{t}p_{[2^{nR_i}]}^U\rVert_1\rightarrow 0,\label{eqn:osrb-main-thm}
\end{equation}
In the above equation, $\mathfrak{B}=\{\mathfrak{B}_i,~i\in[t]\}$ is the set of all random functions and $p_{[2^{nR_i}]}^U$ refers to the uniform distribution on the bin index set $[2^{nR_i}]$. The expectation in \eqref{eqn:osrb-main-thm} is take over random realization of the binning mappings.
\end{theorem}

To prove our results, we state and prove the following improved version of the OSRB theorem which states that not only the average of the total variation distance in \eqref{eqn:osrb-main-thm} converges to zero, but also exponentially fast:
\begin{theorem}\label{thm:errorvanishfast} Assuming that all the random variables in the statement of Theorem \ref{thm:OSRBa0} take values in finite sets, the expected value of the total variation of the joint pmf $P(c^n,\bar{\mathbf{m}})$ from the $p_{c^n}\prod_{i=1}^{t}p_{[2^{nR_i}]}^U$ tends to zero, \emph{exponentially fast} as $2^{-\kappa{n}}$ for some constant $\kappa$, as $n$ approaches infinity.
\end{theorem}
\begin{proof} This follows from the proof of the OSRB theorem (Theorem~1 in \cite{yaar14}) with minor modifications. Here we only mention how the proof should be modified without repeating the entire proof. In our re-statement of the OSRB theorem above, we have used a notation that is suitable for  our purposes here, which is different from the one used in \cite{yaar14}. However, just for the purpose of writing the modification that needs to be made in the proof given in \cite{yaar14}, we adopt the notation and definitions of \cite{yaar14}. We refer the reader to \cite{yaar14} for definition of variables that we use below. 

The proof begins by bounding the total variation distance between two distribution with their fidelity (Lemma 7 of \cite{yaar14}). The paper then states that to show the expected total variation distance goes to zero, it suffices to show that the corresponding expected fidelity term goes to one as $n$ goes to infinity. Now, to show that the total variation distance goes to zero  {exponentially fast} as $2^{-\alpha{n}}$, it suffices to show that the ``one minus the expected fidelity term" goes to zero exponentially fast. This  follows from the fact that if an arbitrary sequence $1-f_n$ tends to zero at least exponentially fast, then $\sqrt{1-f_n^2}=\sqrt{(1-f_n)(1+f_n)}$ also tends to zero exponentially fast.

This fidelity term is bounded from below in equation (104)-(106) as follows:
\begin{align}
&\e \left[F(P(z^n,b_{[1:T]});p(z^n)p^U(b_{[1:T]}))\right]\geq p(\styp)\sqrt{\dfrac{1}{1+\sum_{\emptyset\neq\ms\subseteq\mv}2^{n(R_{\ms}-H(X_{\ms}|Z)+\epsilon)}}}\label{my16}
\end{align}
where $\epsilon$ is an arbitrary positive number and  $\styp$ is the weak typical set defined as follows:
\begin{align}
\styp:=\left\{(x_{[1:T]}^n,z^n): \frac{1}{n}h(x_{[1:T]}^n|z^n)\ge H(X_{[1:T]}|Z)-\epsilon, \right\}.
\end{align}
Now, since $\epsilon$ is fixed, we know that not only probability of i.i.d.\ $X_{[1:T]}^n, Z^n$ being typical converges to one, but it also converges exponentially fast. We also have
\begin{align}\sqrt{\dfrac{1}{1+\sum_{\emptyset\neq\ms\subseteq\mv}2^{n(R_{\ms}-H(X_{\ms}|Z)+\epsilon)}}}&\geq \sqrt{{1-\sum_{\emptyset\neq\ms\subseteq\mv}2^{n(R_{\ms}-H(X_{\ms}|Z)+\epsilon)}}}\\&
\geq {{1-\sum_{\emptyset\neq\ms\subseteq\mv}2^{n(R_{\ms}-H(X_{\ms}|Z)+\epsilon)}}},
\end{align}
converges to one exponentially fast if for each $\ms\subseteq[1:T]$ we have $R_{\ms}< H(X_{\ms}|Z)-\epsilon$. Therefore, both terms on the the right hand side of \eqref{my16} converge to one exponentially fast. Thus, their product also converges to one exponentially fast.

\end{proof}

We also need a linear version of the OSRB theorem. Assume that $M_i$'s are vectors of symbols in a finite field $\mathbb{F}$. Then, $n$ i.i.d.\ repetitions of $M_i$, namely $M_i^n$ can be also understood as a (longer) sequence of symbols in  $\mathbb{F}$. Thus, a linear random binning of rate $R_i$, namely $\mathfrak{B}_i:\mathcal{M}_i^n\rightarrow\bar{\mathcal{M}_i}=\mathbb{F}^{\frac{nR_i}{\log|\mathbb{F}|}}$ can be constructed as $\bar{M}_i={A}_i M^n_i+V_i$ for some random matrices $A_i$ and vectors $V_i$ with mutually independent and uniform entries. We can now state the linear version of the OSRB theorem.
\begin{theorem}[Linear OSRB] Assuming that $M_i$'s are vectors of symbols in a finite field, Theorem \ref{thm:OSRBa0}  holds if we replace the general random binning with linear random binning. 
\end{theorem}
\begin{proof}
The only place where  random binning enters calculation in the proof of the OSRB theorem in \cite{yaar14} are equations (94) and (98) in \cite{yaar14}. But (94) in \cite{yaar14} only uses the uniformity condition which is valid for linear binning (equation \eqref{uniformity}), and (98) in \cite{yaar14} only uses the pairwise independence property that is also valid for linear binning  (equation \eqref{uniformity-pair}). 
\end{proof}

\subsection{Simulation from bin index}
Assume that $X$ is distributed uniformly on some alphabet set, and let $X^n$ be an i.i.d.\ repetitions of $X$. Let $B=\mathfrak{B}(x^n)\in \{0,1, \dots, 2^{nR}-1\}$ be a random binning of $X^n$ at rate $R$. Given any particular realization of the binning, we end up with some joint distribution $p_{BX^n}$ where $B$ is a function of $X^n$. From this joint distribution, we can consider the conditional pmf $p_{X^n|B}$. Observe that multiple $X^n$ may be mapped to $B=b$, hence, $p_{X^n|B}$ is not a deterministic channel. We now ask for the minimum random bit rate required to simulate the channel $p_{X^n|B}$ as defined by Steinberg and Verdu
in 
\cite{SteinbergVerdu}. In other words, given input $B$ of the channel $p_{X^n|B}$, we ask for the minimum number of uniformly random bits (independent of input $B$) that we need to have to be able to accurately simulate the channel $p_{X^n|B}$. In particular, if we denote the simulated channel by $\tilde{p}_{X^n|B}$, we define the total variation distance 
$$\|p_Bp_{X^n|B}-p_B\tilde{p}_{X^n|B}\|_1$$
as a measure of accuracy of channel simulation \cite{SteinbergVerdu}. 

Observe that $H(X^n|B)=H(X^n)-H(B)=n\log|\mathcal{X}|-H(B)\geq n\log|\mathcal{X}|-nR$. Intuitively speaking, to simulate conditional pmf $p_{X^n|B}$, we need a random source of average rate $\log|\mathcal{X}|-R$. The following theorem shows that the rate $\log|\mathcal{X}|-R+\delta$ (for any $\delta>0$) is  sufficient with high probability:

\begin{theorem} \label{thm:thm7} Take some $R<\log|\mathcal{X}|$ and $\delta>0$. Let $T$ be a source of randomness, uniformly distributed over an alphabet $\mathcal{T}$ satisfying $\frac{1}{n}\log|\mathcal{T}|\leq \tilde{R}=\log|\mathcal{X}|-R+\delta$. Given any realization of the binning, a deterministic simulation function $\phi(T, B)$ imposes the channel $$\tilde{p}_{X^n|B}(x^n|b)=\frac{1}{|\mathcal T|}\sum_{t}\mathbbm{1}[\phi(t, b)=x^n],$$
Then, we claim one can find a deterministic simulation function $\phi$ for any realization of the binning such that 
$$\mathbb{E}_{\mathfrak{B}}\|P_BP_{X^n|B}-P_B\tilde{P}_{X^n|B}\|_1\leq 2^{-\eta n}$$
converges to zero exponentially fast in $n$ for some $\eta>0$. Here the expectation is taken over all realizations of the binning. Furthermore, if the  binning from $X^n$ to $B$ is linear, then one can find a deterministic linear simulation function $\phi(T, B)$ satisfying the desired property.
\end{theorem}
\begin{proof}
Fix a  realization of the binning mapping $\mathfrak{B}$. Since $X^n$ is uniformly distributed, the conditional distribution of $X^n$ given $B=b$ is also uniform over the set of sequences $x^n$ that are mapped to $B=b$, \emph{i.e.,} $\{x^n:\mathfrak{B}(x^n)=b\}$. We can successfully simulate $p_{X^n|B=b}$ if we can choose a sequence $x^n$ uniformly at random from the set $\{x^n:\mathfrak{B}(x^n)=b\}$. This would be possible if $|\{x^n:\mathfrak{B}(x^n)=b\}|\leq 2^{n\tilde R}$. Hence, the total variation distance can be bounded from above as follows:
$$\|p_Bp_{X^n|B}-p_B\tilde{p}_{X^n|B}\|_1\leq \sum_b p_B(b)\mathbbm{1}[|\{x^n:\mathfrak{B}(x^n)=b\}|>2^{n\tilde R}],$$
where we used the fact that when $b$ is such that $|\{x^n:\mathfrak{B}(x^n)=b\}|$ is large, the total variation distance can be at most one. 
Thus, by taking average over all random binnings, we have 
\begin{align*}\mathbb{E}_{\mathfrak{B}}\|P_BP_{X^n|B}-P_B\tilde{P}_{X^n|B}\|_1&\leq \mathbb{P}_{\mathfrak{B}, B}[|\{x^n:\mathfrak{B}(x^n)=B\}|>2^{n\tilde R}]
\\&\overset{(a)}{=} \mathbb{P}_{\mathfrak{B}}[|\{x^n:\mathfrak{B}(x^n)=1\}|>2^{n\tilde R}].
\end{align*}
where $(a)$ follows from symmetry. Now, in a random binning, the number of sequences $x^n$ that are mapped to bin index $1$ has a Binomial distribution; we throw $|\mathcal{X}|^n$ sequences and each falls into the first bin with probability $2^{-nR}$. By Markov's inequality, we obtain
$$\mathbb{P}_{\mathfrak{B}}[|\{x^n:\mathfrak{B}(x^n)=1\}|>2^{n\tilde R}]\leq \frac{|\mathcal{X}|^n2^{-nR}}{2^{n\tilde R}}=2^{-n\delta}.$$
Finally, assume that the binning is linear, \emph{i.e.,} $B=AX^n+V$ for some matrices $A$ and $V$. Let the bin index $B$ be a vector of symbols in $\mathbb{F}$ of length $nR'$ where $R'=\frac{R}{\log|\mathbb{F}|}$. 
The set $\{x^n:\mathfrak{B}(x^n)=b\}=\{x^n:Ax^n=b-V\}$ is an affine linear subspace with dimension $\text{Null}(A)=n-\text{Rank}(A)$. This set can be written as $Q(b-V)+NT$ for some matrices $Q$ and $N$, and a uniformly distributed vector $T$ whose length is equal to the dimension of $\text{Null}(A)$. If the rank of $A$ is at least $n(R'-\frac{\delta}{\log|\mathbb{F}|})$, the dimension of the null space will be at most $n(1-R'+\frac{\delta}{\log|\mathbb{F}|})$, and a randomness of size $n(1-R')\log|\mathbb{F}|=n(\log|\mathcal X|-R+\delta)$ would suffice for channel simulation. 
Hence, the total variation distance can be bounded from above as follows:
\begin{align}\mathbb{E}_{\mathfrak{B}}\|P_BP_{X^n|B}-P_B\tilde{P}_{X^n|B}\|_1&\leq \mathbb{P}_{\mathfrak{B}}[\text{Rank}(A)<n(R'-\frac{\delta}{\log|\mathbb{F}|})].\label{eqn:syma0e}
\end{align}
However, for any $R'<1$, it is known that the probability that a random matrix $A_{nR'\times n}$ with uniform entries from $\mathbb{F}$ is not full rank vanishes exponentially fast in $n$; in fact this probability is less than $|\mathbb{F}|^{-n(1-R')}(|\mathbb{F}|-1)^{-1}$ \cite[p.4]{blake}. This completes the proof for the linear case. 

\end{proof}

\end{document}